\PassOptionsToPackage{unicode}{hyperref}
\PassOptionsToPackage{hyphens}{url}
\PassOptionsToPackage{dvipsnames,svgnames,x11names}{xcolor}
\documentclass{article}

\usepackage{arxiv}

\usepackage{amsmath,amssymb}
\usepackage{setspace}
\usepackage{iftex}
\ifPDFTeX
  \usepackage[T1]{fontenc}
  \usepackage[utf8]{inputenc}
  \usepackage{textcomp} 
\else 
  \usepackage{unicode-math}
  \defaultfontfeatures{Scale=MatchLowercase}
  \defaultfontfeatures[\rmfamily]{Ligatures=TeX,Scale=1}
\fi
\usepackage{lmodern}
\ifPDFTeX\else  
\fi
\IfFileExists{upquote.sty}{\usepackage{upquote}}{}
\IfFileExists{microtype.sty}{
  \usepackage[]{microtype}
  \UseMicrotypeSet[protrusion]{basicmath} 
}{}
\makeatletter
\@ifundefined{KOMAClassName}{
  \IfFileExists{parskip.sty}{%
    \usepackage{parskip}
  }{
    \setlength{\parindent}{0pt}
    \setlength{\parskip}{6pt plus 2pt minus 1pt}}
}{
  \KOMAoptions{parskip=half}}
\makeatother
\usepackage{xcolor}
\ifx\paragraph\undefined\else
  \let\oldparagraph\paragraph
  \renewcommand{\paragraph}[1]{\oldparagraph{#1}\mbox{}}
\fi
\ifx\subparagraph\undefined\else
  \let\oldsubparagraph\subparagraph
  \renewcommand{\subparagraph}[1]{\oldsubparagraph{#1}\mbox{}}
\fi

\usepackage{color}
\usepackage{fancyvrb}

\DefineVerbatimEnvironment{Highlighting}{Verbatim}{commandchars=\\\{\}}
\usepackage{framed}
\definecolor{shadecolor}{RGB}{241,243,245}
\newenvironment{Shaded}{\begin{snugshade}}{\end{snugshade}}

\newcommand{\AttributeTok}[1]{\textcolor[rgb]{0.40,0.45,0.13}{#1}}

\newcommand{\ConstantTok}[1]{\textcolor[rgb]{0.56,0.35,0.01}{#1}}

\newcommand{\DecValTok}[1]{\textcolor[rgb]{0.68,0.00,0.00}{#1}}

\newcommand{\FunctionTok}[1]{\textcolor[rgb]{0.28,0.35,0.67}{#1}}

\newcommand{\NormalTok}[1]{\textcolor[rgb]{0.00,0.23,0.31}{#1}}

\newcommand{\StringTok}[1]{\textcolor[rgb]{0.13,0.47,0.30}{#1}}

\providecommand{\tightlist}{%
  \setlength{\itemsep}{0pt}\setlength{\parskip}{0pt}}\usepackage{longtable,booktabs,array}
\usepackage{calc} 
\usepackage{etoolbox}
\makeatletter
\patchcmd\longtable{\par}{\if@noskipsec\mbox{}\fi\par}{}{}
\makeatother
\IfFileExists{footnotehyper.sty}{\usepackage{footnotehyper}}{\usepackage{footnote}}
\makesavenoteenv{longtable}
\usepackage{graphicx}
\makeatletter
\def\maxwidth{\ifdim\Gin@nat@width>\linewidth\linewidth\else\Gin@nat@width\fi}
\def\maxheight{\ifdim\Gin@nat@height>\textheight\textheight\else\Gin@nat@height\fi}
\makeatother
\setkeys{Gin}{width=\maxwidth,height=\maxheight,keepaspectratio}
\makeatletter
\def\fps@figure{htbp}
\makeatother
\NewDocumentCommand\citeproctext{}{}

\makeatletter
 \let\@cite@ofmt\@firstofone
 \def\@biblabel#1{}
 \def\@cite#1#2{{#1\if@tempswa , #2\fi}}
\makeatother
\newlength{\cslhangindent}
\newlength{\csllabelwidth}
\newenvironment{CSLReferences}[2] 
 {\begin{list}{}{%
  \setlength{\itemindent}{0pt}
  \setlength{\leftmargin}{0pt}
  \setlength{\parsep}{0pt}
  \ifodd #1
   \setlength{\leftmargin}{\cslhangindent}
   \setlength{\itemindent}{-1\cslhangindent}
  \fi
  \setlength{\itemsep}{#2\baselineskip}}}
 {\end{list}}
\usepackage{calc}

\usepackage{isomath}
\usepackage{booktabs}
\usepackage{longtable}
\usepackage{array}
\usepackage{multirow}
\usepackage{wrapfig}
\usepackage{float}
\usepackage{colortbl}
\usepackage{pdflscape}
\usepackage{tabu}
\usepackage{threeparttable}
\usepackage{threeparttablex}
\usepackage[normalem]{ulem}
\usepackage{makecell}
\usepackage{xcolor}
\usepackage{orcidlink}
\usepackage{amsmath}
\usepackage{amsthm}
\usepackage{float}
\usepackage{hyperref}
\usepackage[utf8]{inputenc}
\usepackage{bm}
\def\tightlist{}
\usepackage{setspace}
\newcommand\pD{$p\text{-}D$}

\newcommand\gD{$2\text{-}D$}
\newcommand\bD{$1\text{-}D$}
\newcommand\fD{$4\text{-}D$}
\newtheorem*{theorem}{Theorem}
\makeatletter
\@ifpackageloaded{caption}{}{\usepackage{caption}}
\AtBeginDocument{%
\ifdefined\contentsname
  \renewcommand*\contentsname{Table of contents}
\else
  \newcommand\contentsname{Table of contents}
\fi
\ifdefined\listfigurename
  \renewcommand*\listfigurename{List of Figures}
\else
  \newcommand\listfigurename{List of Figures}
\fi
\ifdefined\listtablename
  \renewcommand*\listtablename{List of Tables}
\else
  \newcommand\listtablename{List of Tables}
\fi
\ifdefined\figurename
  \renewcommand*\figurename{Figure}
\else
  \newcommand\figurename{Figure}
\fi
\ifdefined\tablename
  \renewcommand*\tablename{Table}
\else
  \newcommand\tablename{Table}
\fi
}
\@ifpackageloaded{float}{}{\usepackage{float}}
\floatstyle{ruled}
\@ifundefined{c@chapter}{\newfloat{codelisting}{h}{lop}}{\newfloat{codelisting}{h}{lop}[chapter]}
\floatname{codelisting}{Listing}

\makeatother
\makeatletter
\@ifpackageloaded{caption}{}{\usepackage{caption}}
\@ifpackageloaded{subcaption}{}{\usepackage{subcaption}}
\makeatother
\ifLuaTeX
  \usepackage{selnolig}  
\fi
\usepackage{bookmark}

\IfFileExists{xurl.sty}{\usepackage{xurl}}{} 
\title{Is this normal? A new projection pursuit index to assess a sample
against a multivariate null distribution}

\author{
    Annalisa Calvi
   \\
    School of Mathematics \\
    Monash University \\
   \\
  \texttt{\href{mailto:annalisa.calvi@monash.edu}{\nolinkurl{annalisa.calvi@monash.edu}}} \\
   \And
    Ursula Laa
   \\
    Institute of Statistics \\
    University of Natural Resources and Life Sciences Vienna \\
   \\
  \texttt{\href{mailto:ursula.laa@boku.ac.at}{\nolinkurl{ursula.laa@boku.ac.at}}} \\
   \And
    Dianne Cook
   \\
    Department of Econometrics and Business Statistics \\
    Monash University \\
   \\
  \texttt{\href{mailto:dicook@monash.edu}{\nolinkurl{dicook@monash.edu}}} \\
  }

\begin{document}
\maketitle
\begin{abstract}
Many data problems contain some reference or normal conditions, upon
which to compare newly collected data. This scenario occurs in data
collected as part of clinical trials to detect adverse events, or for
measuring climate change against historical norms. The data is typically
multivariate, and often the normal ranges are specified by a
multivariate normal distribution. The work presented in this paper
develops methods to compare the new sample against the reference
distribution with high-dimensional visualisation. It uses a projection
pursuit guided tour to produce a sequence of low-dimensional projections
steered towards those where the new sample is most different from the
reference. A new projection pursuit index is defined for this purpose.
The tour visualisation also includes drawing of the projected ellipse,
which is computed analytically, corresponding to the reference
distribution. The methods are implemented in the R package,
\texttt{tourr}.
\end{abstract}

\setstretch{2}
\section{Introduction}\label{introduction}

Linear projections are useful in many aspects of statistical analysis of
multivariate data, and especially useful for visualising the data. A
linear projection provides a dimension reduction while maintaining
interpretability. For example, a biplot (Gabriel 1971; Gower and Hand
1996) shows the structure creating the maximum variance in the data, and
also visualizing the projection matrix to learn which variables
contribute to it. We might find clusters of outliers that were hiding in
high dimensions.

More generally, projection pursuit (Friedman and Tukey 1974; Jones and
Sibson 1987; Huber 1985) defines a quantitative criterion for the
interestingness of a projection (a projection pursuit index), and
searches the space of possible projections for the most interesting one
to display. We can also define sequences of interpolated linear
projections to better understand a multivariate distribution. Animating
a randomly selected interpolated sequence of linear projections is
called a grand tour (Asimov 1985; Buja and Asimov 1986; Buja et al.
2005; Cook et al. 2006; Lee et al. 2022). The combination of these two
approaches would then use a projection pursuit index to select
interesting projections, but display them via an interpolated path to
provide context. This is called a guided tour (Cook et al. 1995).

The question is whether we can use these techniques to assess new data
samples in the context of an established normal, such as a specific
multivariate normal distribution. In physics, the normal distribution
may describe experimental results, or a global fit for a selected model,
and we might want to compare to sets of other models. In medical
applications, the normal distribution might summarize historic data of a
healthy population and we compare it to samples from new patients. In
outlier detection we might use robust measures to define the normal
distribution and look for anomalies.

This paper describes a new projection pursuit index to find projections
where a new sample is most distant from the existing normal
distribution, that is integrated with the guided tour to visualise the
ways that the data departs from the reference distribution. This work
also expands on current outlier detection methods such as
Kandanaarachchi and Hyndman (2021) and Loperfido (2018) by providing
integrated multivariate visualisation to examine the outlyingness in
different projections. It also provides the ability to examine potential
outliers interactively - when they are outlying in different ways it can
interfere with detection numerically, but visually one might group
different types of outlyingness, and operate the outlier detection on
subsets.

The paper is organised as follows. Section~\ref{sec-background} provides
more context for the methods and visualisation.
Section~\ref{sec-anomaly-index} provides the details of the new index,
and rendering projections of the reference ellipse.
Section~\ref{sec-implementation} describes the implementation. Two
examples, one medical and one on climate extremes, illustrate usage in
Section~\ref{sec-examples}.

\section{Background}\label{sec-background}

Following the ideas expressed in Cook and Swayne (2007), to compare a
new sample with an existing norm, like a multivariate normal
distribution, in higher than two dimensions, you would use two samples
of points. The norm is represented by points on the surface of a
\(p\)-dimensional ellipsoid, with size corresponding to a chosen level
curve of the density function of the reference distribution. This can be
generated using these steps:

\begin{enumerate}
\def\labelenumi{\arabic{enumi}.}
\tightlist
\item
  Simulate a sample of observations (\(\mathbfit{x}\), which are \pD{}
  vectors) from \(N_p(\mathbfit{\mu}, \Sigma)\).
\item
  Transform each observation to have unit distance from the mean,
  \(\frac{\mathbfit{x}^\top}{||\mathbfit{x}^\top||}\). The sample is now
  uniform on the surface of the hypersphere.
\item
  Transform the shape from a sphere to an ellipsoid using the specific
  variance-covariance.
\item
  Shift the resulting sample to center it on the mean vector, and resize
  to correspond to the desired density level. This could also be treated
  like a confidence ellipse, where observations outside are more than,
  say, three standard deviations from the sample mean.
\end{enumerate}

New observations can be visually compared with this ellipsoid by
plotting them together using a tour, making sure that any
transformations used on the reference distribution are also applied
during pre-processing. Figure~\ref{fig-ci} illustrates this process for
\gD.

\begin{figure}

\centering{

\includegraphics[width=1\textwidth,height=\textheight]{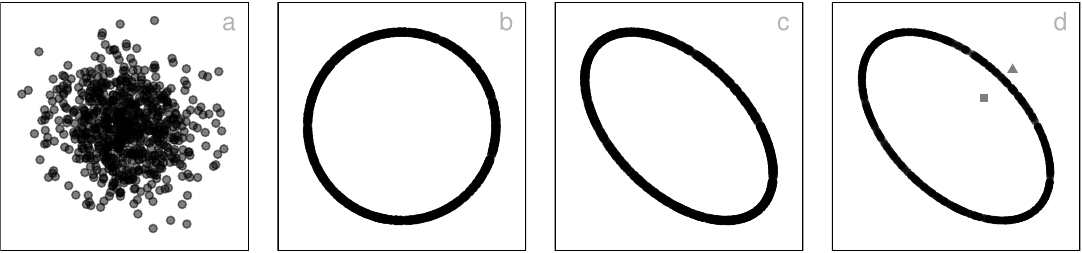}

}

\caption{\label{fig-ci}Simulating a uniform sample on a sphere involves
sampling from a multivariate normal (a) and transforming each
observation to have length equal to 1 (b). A 95\% confidence ellipsoid
is generated by transforming the sphere relative to a specified
variance-covariance matrix, and sizing using a critical value (c), and
new observations (triangle, square) can be visually assessed to be
inside or outside by plotting with the ellipsoid (d).}

\end{figure}%

\begin{figure}

\begin{minipage}{0.50\linewidth}

\centering{

\includegraphics[width=2.60417in,height=\textheight]{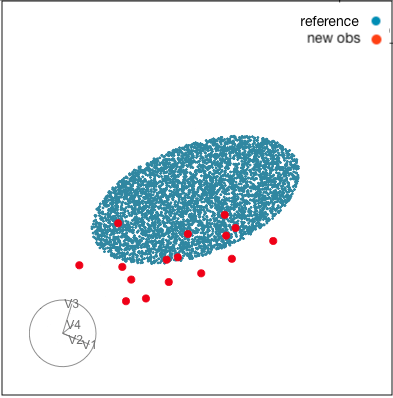}

}

\subcaption{\label{fig-compare1}Reference ellipse}

\end{minipage}%
\begin{minipage}{0.50\linewidth}

\centering{

\includegraphics[width=2.60417in,height=\textheight]{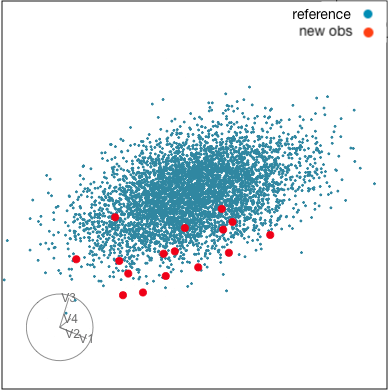}

}

\subcaption{\label{fig-compare}Reference sample}

\end{minipage}%

\caption{\label{fig-compare}Illustration of existing procedures for a
\(4\text{-}D\) problem: compare new sample with (a) points on the
surface of a \(p\text{-}D\) ellipse representing the reference
distribution, (b) a sample of observations simulated from the reference
distribution. A single \(2\text{-}D\) projection from a tour. (The
circle and projected axes represent projection coefficients, and
animated gifs of these tours are available in the supplementary
material.) Both approaches are useful, but need the refinements
described in this paper.}

\end{figure}%

This approach for comparing new observations with a reference
distribution can be applied to any \pD{} problem.
Figure~\ref{fig-compare} illustrates the procedures for a \fD{} problem.
The new sample of observations is compared against the reference
distribution in two ways (a) against a \fD{} ellipse corresponding to a
level curve of density function of the reference distribution and (b)
against a sample simulated from the reference distribution. Here we show
a single \gD{} projection from a tour of the full \pD{}. The circle and
line segments indicate the coefficients of the \(4\times 2\) basis used
to project the data in these plots. The difference between the new
sample of points (dark/red) from the sample is in the bottom left to top
right direction, corresponding mostly to the V3 axis direction. Roughly
orthogonal to this is a combination of V1 and V2, so the difference
between new and reference as indicated by this projection is due to V4
contrasted with V1 and V2. The ellipse makes it clearer but the sample
can be useful, too, particularly if the reference distribution is not a
multivariate normal distribution. (The supplementary material includes
animated gifs showing these data in a tour, allowing one to see many
different projections of the \fD{} points.)

While this is a flexible and useful approach, there are two problems:
(1) the ragged edges of the points marking the projected ellipse make it
difficult to compare the new sample precisely against acceptable ranges,
and (2) there is no projection pursuit index that can guide the tour
towards the directions (projections) where the samples are most
different from the reference.

What would be better, and what is developed in this paper, is to
analytically define the confidence ellipsoid and display the projected
ellipsoid as a geometric shape rather than a sample of points. Further,
a new projection pursuit index that is based on flagging observations
that are outside, and would steer the tour to projections that reveal
the extent of the difference. These new procedures are described in the
next section.

\section{Anomaly index}\label{sec-anomaly-index}

\subsection{Projecting an ellipsoid}\label{projecting-an-ellipsoid}

Let \(\mathbfit{x}\) be a \pD{} vector. A \pD{} ellipsoid corresponding
to a given variance-covariance (\(\Sigma\)) and mean vector
(\(\mathbfit{\mu}\)) is described by the equation

\begin{equation}
(\mathbfit{x}-\mathbfit{\mu}) \Sigma^{-1}(\mathbfit{x}-\mathbfit{\mu})^T = c^2
\label{eq:pd}
\end{equation}

where \(c\) is a constant that depends on a specific confidence level.

Let \(P\) be a (\(p\times 2\)) matrix, whose columns form an orthonormal
basis for the \gD{} projection space. Let \(\mathbfit{y}\) be a \gD{}
vector, and let \(\mathbfit{\mu}_p := \mathbfit{\mu} P\) denote the
projected mean.

\begin{theorem}
  The projection of the \pD{} ellipsoid described in Equation~\ref{eq:pd} onto the \gD{} projection space described by $P$ has the equation
\begin{equation}
(\mathbfit{y} - \mathbfit{\mu}_p)(P^T \Sigma P)^{-1}(\mathbfit{y} - \mathbfit{\mu}_p)^T = c^2.
\end{equation}
\end{theorem}
\begin{proof}
The projection of a hyperellipsoid onto \gD{} space is an ellipse. The boundary of this \gD{} ellipse is the orthogonal projection of select points on the boundary of the hyperellipsoid. For $\mathbfit{x}$ on the hyperellipsoid, its orthogonal projection $\mathbfit{x}P$ is on the boundary of the projected ellipse if and only if the tangent plane to the hyperellipsoid at $\mathbfit{x}$ is perpendicular to the projection plane described by $P$. As the hyperellipsoid is a level curve, described by Equation~(\ref{eq:pd}), this means the gradient of the left hand side of Equation~(\ref{eq:pd}) at $\mathbfit{x}$ is parallel to the projection plane. Note that any vector parallel to the projection plane can be expressed as $2\mathbfit{s}P^T$ for some $\mathbfit{s} \in \mathbb{R}^2$.

Therefore, the boundary of the ellipse consists of the points $\mathbfit{x}P$ for which $\mathbfit{x}$ satisfies

\begin{equation}
2\mathbfit{s}P^T = \nabla \left[(\mathbfit{x}-\mathbfit{\mu}) \Sigma^{-1}(\mathbfit{x}-\mathbfit{\mu})^T\right] = 2 (\mathbfit{x}-\mathbfit{\mu}) \Sigma^{-1}
\label{eq:grad}
\end{equation}

for some $\mathbfit{s} \in \mathbb{R}^2$. Rearranging this equation gives $\mathbfit{x}-\mathbfit{\mu} = \mathbfit{s} P^T \Sigma$.
Making this substitution in Equation~(\ref{eq:pd}) yields

\begin{equation}
c^2 = (\mathbfit{x}-\mathbfit{\mu}) \Sigma^{-1}(\mathbfit{x}-\mathbfit{\mu})^T = \mathbfit{s} P^T \Sigma P \mathbfit{s}^T
\label{eq:c2}
\end{equation}

We denote points on the boundary of the projected ellipse by $\mathbfit{y}$, so that $\mathbfit{y} = \mathbfit{x}P$ for some $\mathbfit{x}$ satisfying Equation~(\ref{eq:grad}). Then, as $\mathbfit{\mu}_p = \mathbfit{\mu} P$ by definition, $\mathbfit{y} - \mathbfit{\mu}_p = (\mathbfit{x} - \mathbfit{\mu})P = \mathbfit{s} P^T \Sigma P$, again substituting $\mathbfit{x}-\mathbfit{\mu} = \mathbfit{s} P^T \Sigma$. We then substitute $(\mathbfit{y} - \mathbfit{\mu}_p) (P^T \Sigma P)^{-1}$ for $\mathbfit{s}$ in Equation~(\ref{eq:c2}). From this we can compute the analog equation for the projection as

\begin{equation}
(\mathbfit{y} - \mathbfit{\mu}_p)(P^T \Sigma P)^{-1}(\mathbfit{y} - \mathbfit{\mu}_p)^T = c^2
\end{equation}

as claimed.
\end{proof}

This means the matrix \((P^T \Sigma P)^{-1}\) defines the ellipse in the
\gD{} projection. In general \(c\) could be any constant, but typically
we would select it as a quantile of the \(\chi^2\) distribution, so that
the size of the ellipse corresponds to a selected probability.

\subsection{Index specification}\label{index-specification}

Mahalanobis distance (Mahalanobis 1936) measures the distance from the
center \(\mathbfit{\mu}\) in units of standard deviations, and is
computed using the variance-covariance matrix \(\Sigma\).

\begin{equation}
d_M (\mathbfit{x}) = \sqrt{(\mathbfit{x}-\mathbfit{\mu}) \Sigma^{-1}(\mathbfit{x}-\mathbfit{\mu})^T}
\end{equation}

To define a measure of an interesting projection we propose to use the
average Mahalanobis distance in the projection, for a subset of points,
\(W\). The set of points could be chosen in different ways, but the
default is those that are outside the specified ellipsoid in \pD{},
\(W = \{\mathbfit{x}: (\mathbfit{x}-\mathbfit{\mu}) \Sigma^{-1}(\mathbfit{x}-\mathbfit{\mu})^T > c^2\}\).
Alternatives could be to select a set of observations with the largest
Mahalanobis distance, manually select observations or possibly a group
of points identified via clustering of the extremes.

We define our new index as

\begin{equation}
\sum_{\mathbfit{w} \in W} (\mathbfit{w} - \mathbfit{\mu}) P (P^T\Sigma P)^{-1}P^T(\mathbfit{w} - \mathbfit{\mu})^T.
\end{equation}

\subsection{Additional considerations}\label{additional-considerations}

If the observations in \(W\) are primarily departing from the normal
range in the same direction, the index will be expected to perform well
in finding this average direction. However, if the observations have
very different departures from the norm, it may be useful to break them
into groups, and separately optimize on these subsets. One could
consider clustering these observations using angular distance to find
groups of observations that have similar directions of departure. This
is illustrated in the second example, see Section~\ref{sec-weather}.

\section{Implementation}\label{sec-implementation}

This is implemented in the \texttt{tourr} (Wickham et al. 2011, 2024)
package, where the projected ellipsoid can be drawn for each projection.
The guided tour will take arguments specifying the data, and the null
variance-covariance matrix, as shown in the example code.

\begin{Shaded}
\begin{Highlighting}[]
\FunctionTok{library}\NormalTok{(tourr)}
\FunctionTok{animate\_xy}\NormalTok{(samp, }\FunctionTok{guided\_anomaly\_tour}\NormalTok{(}\FunctionTok{anomaly\_index}\NormalTok{(),}
  \AttributeTok{ellipse=}\NormalTok{vc\_null), }\AttributeTok{ellipse=}\NormalTok{vc\_null, }
  \AttributeTok{axes =} \StringTok{"bottomleft"}\NormalTok{, }\AttributeTok{half\_range=}\DecValTok{5}\NormalTok{, }\AttributeTok{center=}\ConstantTok{FALSE}\NormalTok{)}
\end{Highlighting}
\end{Shaded}

A random projection as well as the final view after optimization are
shown in Figure Figure~\ref{fig-anomaly}. The optimal projection (b) can
be used to understand what combination of variables is driving the
difference from the reference distribution.

\begin{figure}

\begin{minipage}{0.50\linewidth}

\centering{

\includegraphics[width=2.60417in,height=\textheight]{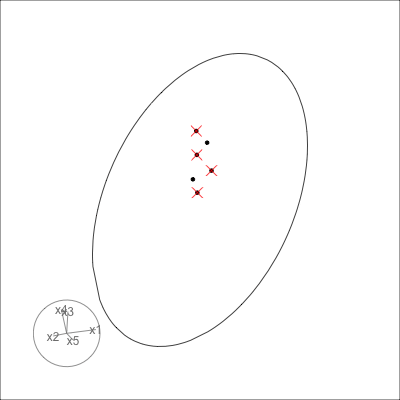}

}

\subcaption{\label{fig-anomaly1}Random projection}

\end{minipage}%
\begin{minipage}{0.50\linewidth}

\centering{

\includegraphics[width=2.60417in,height=\textheight]{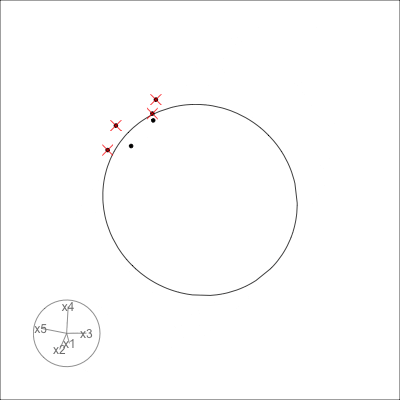}

}

\subcaption{\label{fig-anomaly2}Optimal projection}

\end{minipage}%

\caption{\label{fig-anomaly}Two projections of simulated example data
corresponding to the sample code: (a) random projection where sample is
inside the 2-D ellipse, (b) optimal projection from index, showing most
of the sample outside. A red cross indicates that the point is outside
the p-D ellipse. The optimal projection uses mostly variables
\(x_4, x_5\), which is expected because these are the two directions
where the sample most differs from the norm.}

\end{figure}%

To optimize the index any optimization algorithm implemented in the
\texttt{tourr} package can be used. The focus here is on the final
projection identified, and direct interpretation of index values is not
considered. It may however be interesting to compare different
optimization strategies and compare the index values along the
optimization paths, as available in Zhang et al. (2021).

\section{Examples}\label{sec-examples}

\subsection{Health: liver function
tests}\label{health-liver-function-tests}

This example is motivated by a problem posed during consulting with a
pharmaceutical company. The data shown here is simulated, for privacy
purposes, but reflects the patterns seen in the clinical data. Liver
function tests commonly provide measurements on albumin, protein,
bilirubin, gamma-glutamyl-transferase (GGT), aspartate aminotransferase
(AST), alkaline phosphatase (ALP) and alanine aminotransferase (ALT).
There are normal ranges on these measurements reported by Lala, Zubair,
and Minter (2023) and listed in Table~\ref{tbl-liver-norms}.

\begin{longtable}[]{@{}
  >{\raggedright\arraybackslash}p{(\columnwidth - 14\tabcolsep) * \real{0.0421}}
  >{\raggedleft\arraybackslash}p{(\columnwidth - 14\tabcolsep) * \real{0.1474}}
  >{\raggedleft\arraybackslash}p{(\columnwidth - 14\tabcolsep) * \real{0.1474}}
  >{\raggedleft\arraybackslash}p{(\columnwidth - 14\tabcolsep) * \real{0.2000}}
  >{\raggedleft\arraybackslash}p{(\columnwidth - 14\tabcolsep) * \real{0.1158}}
  >{\raggedleft\arraybackslash}p{(\columnwidth - 14\tabcolsep) * \real{0.1158}}
  >{\raggedleft\arraybackslash}p{(\columnwidth - 14\tabcolsep) * \real{0.1158}}
  >{\raggedleft\arraybackslash}p{(\columnwidth - 14\tabcolsep) * \real{0.1158}}@{}}

\caption{\label{tbl-liver-norms}Normal ranges provided for liver
function tests.}

\tabularnewline

\toprule\noalign{}
\begin{minipage}[b]{\linewidth}\raggedright
\end{minipage} & \begin{minipage}[b]{\linewidth}\raggedleft
albumin (g/L)
\end{minipage} & \begin{minipage}[b]{\linewidth}\raggedleft
protein (g/L)
\end{minipage} & \begin{minipage}[b]{\linewidth}\raggedleft
bilirubin (µmol/L)
\end{minipage} & \begin{minipage}[b]{\linewidth}\raggedleft
GGT (IU/L)
\end{minipage} & \begin{minipage}[b]{\linewidth}\raggedleft
AST (IU/L)
\end{minipage} & \begin{minipage}[b]{\linewidth}\raggedleft
ALP (IU/L)
\end{minipage} & \begin{minipage}[b]{\linewidth}\raggedleft
ALT (IU/L)
\end{minipage} \\
\midrule\noalign{}
\endhead
\bottomrule\noalign{}
\endlastfoot
min & 35 & 60 & 0 & 2 & 5 & 30 & 5 \\
max & 50 & 80 & 20 & 44 & 30 & 120 & 40 \\

\end{longtable}

These measurements are also likely correlated, based on guidance like
the \emph{ratio of AST to ALT of 2:1 indicates possible alcohol abuse}.
When a correlation matrix for normal patients is provided the null
ellipse can be computed to be used to examine new samples.

Figure~\ref{fig-liver} illustrates two examples. The first is similar to
the consulting project. A sample of liver test scores for new patients
was provided in order to examine their values relative to the normal
range. Here only four tests are used, GGT, AST, ALP, ALT. Plot (a) shows
a projection of this sample relative to the normal range, using the
anomaly index guided tour. Some of the patients are outside the 4D
confidence ellipse but all of the patients are off center, located away
from the mean. What has been typical in the past is to compute normal
values based on tests of healthy young males. Although this example is
simulated, it matches what was learned during the project, that the
samples examined were all recorded on women, and they were
systematically different from the normal.

\begin{figure}

\begin{minipage}{0.50\linewidth}

\centering{

\includegraphics[width=2.60417in,height=\textheight]{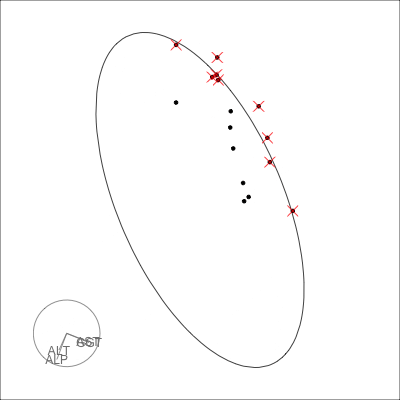}

}

\subcaption{\label{fig-liver1}Sample of female patients}

\end{minipage}%
\begin{minipage}{0.50\linewidth}

\centering{

\includegraphics[width=2.60417in,height=\textheight]{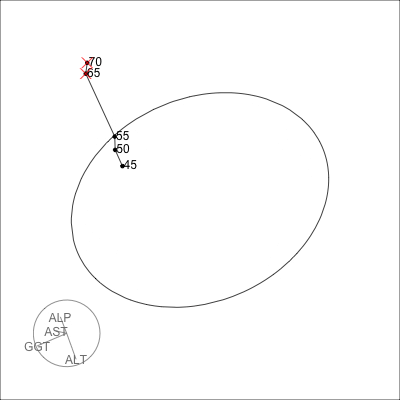}

}

\subcaption{\label{fig-liver2}Longitudinal profile of an aging patient}

\end{minipage}%

\caption{\label{fig-liver}Two projections of simulated example data
corresponding to common liver tests: (a) sample of female patients, (b)
longitudinal test results for a single patient. Red cross indicates
observation is outside the 4-D confidence ellipse. The female patients,
the whole group (a), are systematically different from the normal range.
As the patient ages (b), the level of ALP increases and ALT decreases.}

\end{figure}%

The second example shows a longitudinal record of a single patient,
measured at ages 45, 50, 55, 65 and 70. The lines connect the records in
time order. The projection corresponds to the maximum from a
projection-pursuit guided tour using the anomaly index.

From the axes representation in Figure~\ref{fig-liver} (b) of this
projection, we can see that the direction that profile extends is
primarily contrasting ALT (fourth row) and ALP (third row). This is
consistent with aging, where ALP increases and ALT decreases. Although
this is simulated data, it matches what was observed in the clinical
data that ALP and ALT test scores systematically change with age.

These two examples illustrate how visualizing the data using the anomaly
guided tour can help to understand differences between groups of
correlated test scores.

\subsection{Robust statistics: weather extremes}\label{sec-weather}

This example is motivated by the weather data example from Mayrhofer and
Filzmoser (2023). We illustrate using the anomaly index to compare
potential outliers with a reference normal distribution that is derived
using robust methods on the original data. The data contains 16
numerical measurements that are averages across the three summer months
June, July and August, reported for 68 years (1955 - 2022) (see
Mayrhofer and Filzmoser (2023) for more details).

To estimate the underlying normal distribution the data is first
centered and scaled using the median and the median absolute deviation
(MAD), before applying the minimum covariance determinant (MCD)
estimator (Rousseeuw 1985) using the implementation in Maechler et al.
(2024). The MCD estimates for the mean vector and the
variance-covariance matrix are then used to define the reference normal
distribution.

Here we will consider points to be outlying if they are more than
\(5 \sigma\) away from that mean value. With \(p=16\) this corresponds
to a Mahalanobis distance of \(60\) or larger. This will identify \(20\)
out of the \(68\) observations as outlying. Since these are outlying in
different combinations of variables, as found in Mayrhofer and Filzmoser
(2023), we will further separate the outlying points into clusters based
on similarity in direction. The similarity is computed by first
normalizing observations to have length \(1\) and then apply k-means
clustering with Euclidean distance. We use the Dunn index (Halkidi,
Batistakis, and Vazirgiannis 2002), computed using Hennig (2024), to
select the preferred number of clusters, \(k=5\).

The new anomaly index is first applied to the full dataset, such that
the final projection will be found by averaging distance of points in
many directions, see Figure~\ref{fig-weather} (a). It provides a global
picture showing where the outlying points differ from the normal
distribution. We can see that the contrast between maximum (atmx) and
minimum (atmn) temperature is primarily used in the direction where
points are most outlying. In the orthogonal direction the variables
precipitation (prc) and number of days with maximum temperature above
25°C (nsm) can be used to separate the clusters 2, 3 and 5. With this
solution the two points in cluster 4 cannot be resolved, they fall
inside the ellipse. To better understand what is different for those two
years we separately run the anomaly index for this subset, see
Figure~\ref{fig-weather} (b). From the result we see that the maximum
temperature is not relevant for this cluster, but that they have low
values for the minimum temperature. We also see that there is still some
averaging, and the preferred solution is not placing either point on the
outside of the projected ellipse.

\begin{figure}

\begin{minipage}{0.50\linewidth}

\centering{

\includegraphics[width=2.86458in,height=\textheight]{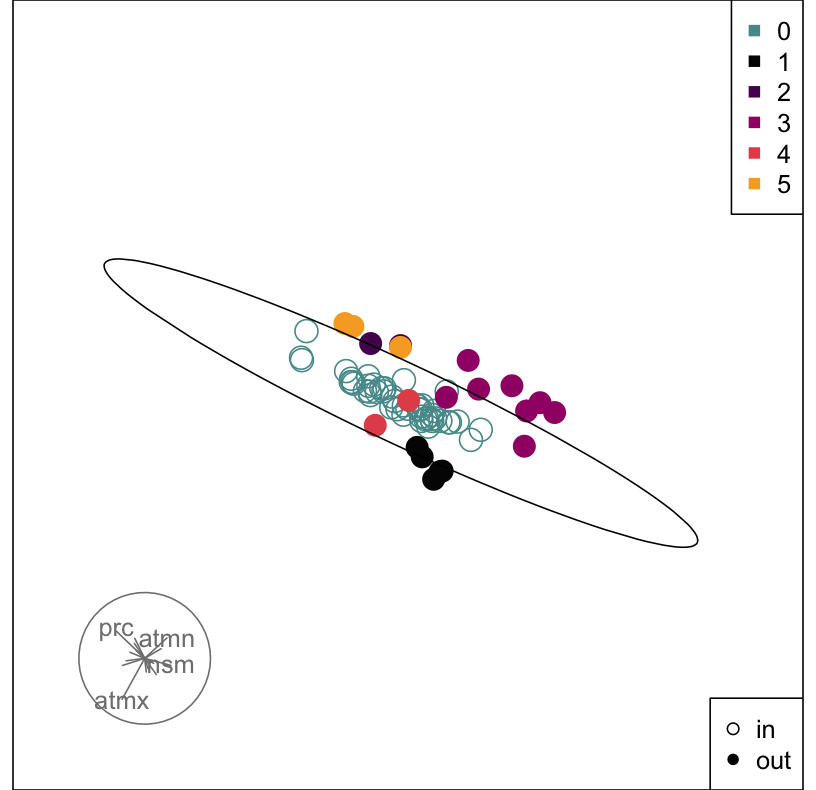}

}

\subcaption{\label{fig-weather1}All outlying points clustered}

\end{minipage}%
\begin{minipage}{0.50\linewidth}

\centering{

\includegraphics[width=2.86458in,height=\textheight]{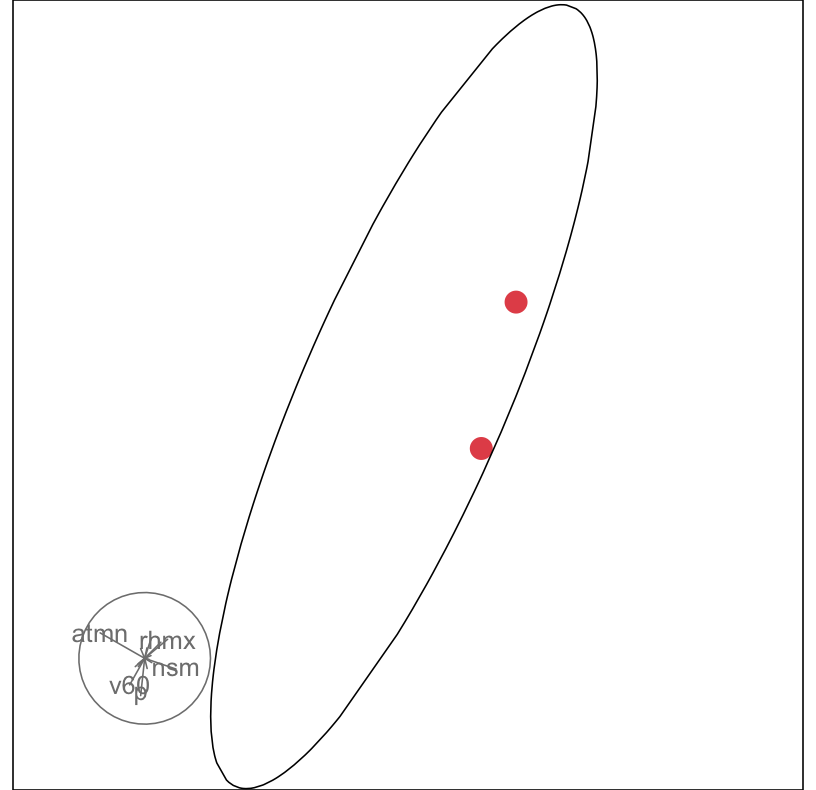}

}

\subcaption{\label{fig-weather2}Focus on cluster 4}

\end{minipage}%

\caption{\label{fig-weather}Examining weather anomalies using the
anomaly index guided tour to assess what combinations of variables are
creating the extremes. Anomalies were grouped based on similarity in
direction, and colour indicates cluster membership. In (a) the anomaly
index was computed on all outlying points, and most clusters (except for
cluster 4) are outlying points are somewhat orthogonal to the normal
ellipse. This extremeness is primarily due to temperature min
(\texttt{atmn}) and max (\texttt{atmx}) variables as read from the axes.
In (b) only the two points in cluster 4 are used for computing the
index. These two observations are primarily extreme in the minimum
temperature (\texttt{atmn}) variable, and have lower than normal
minima.}

\end{figure}%

To compare our results to what is obtained by outlier detection methods,
we reproduce one of the results from Mayrhofer and Filzmoser (2023) in
Figure~\ref{fig-shapley}. On top we can see the cluster assignments for
all points that we have identified as outliers using the
\(p\)-dimensional Mahalanobis distance.

For example we see that cluster 4 contains the years 1962 and 1996, and
for these two years the result from Mayrhofer and Filzmoser (2023) also
indicates unusually low values for the minimum temperature (atmn), while
there is disagreement between outlyingness in other variables. Note that
compared to the years 1965 and 1966 they have unusually low minimum
temperature (atmn) while having average values for the number of days
with maximum temperature above 25°C (nsm), and this contrast is also
captured in the projection.

When looking at the overall picture (Figure~\ref{fig-weather} (a)), we
can see that the linear projection is showing the difference between
minimum (atmn) and maximum temperature (atmx), which is where clusters
are outlying in different directions. The direction of cluster 1 is
where the average maximum temperature (atmx) is unusually higher
compared to the minimum temperature (atmn), while clusters 3 corresponds
to years where the minimum temperature is unusually close to the maximum
temperature. This combined information cannot be resolved with the cell
detection algorithm, but we see that for example the most recent years
(cluster 3) were tagged for having unusually high values for both the
minimum (atmn) and the maximum temperature (atmx).

\begin{figure}

\centering{

\includegraphics[width=5.72917in,height=\textheight]{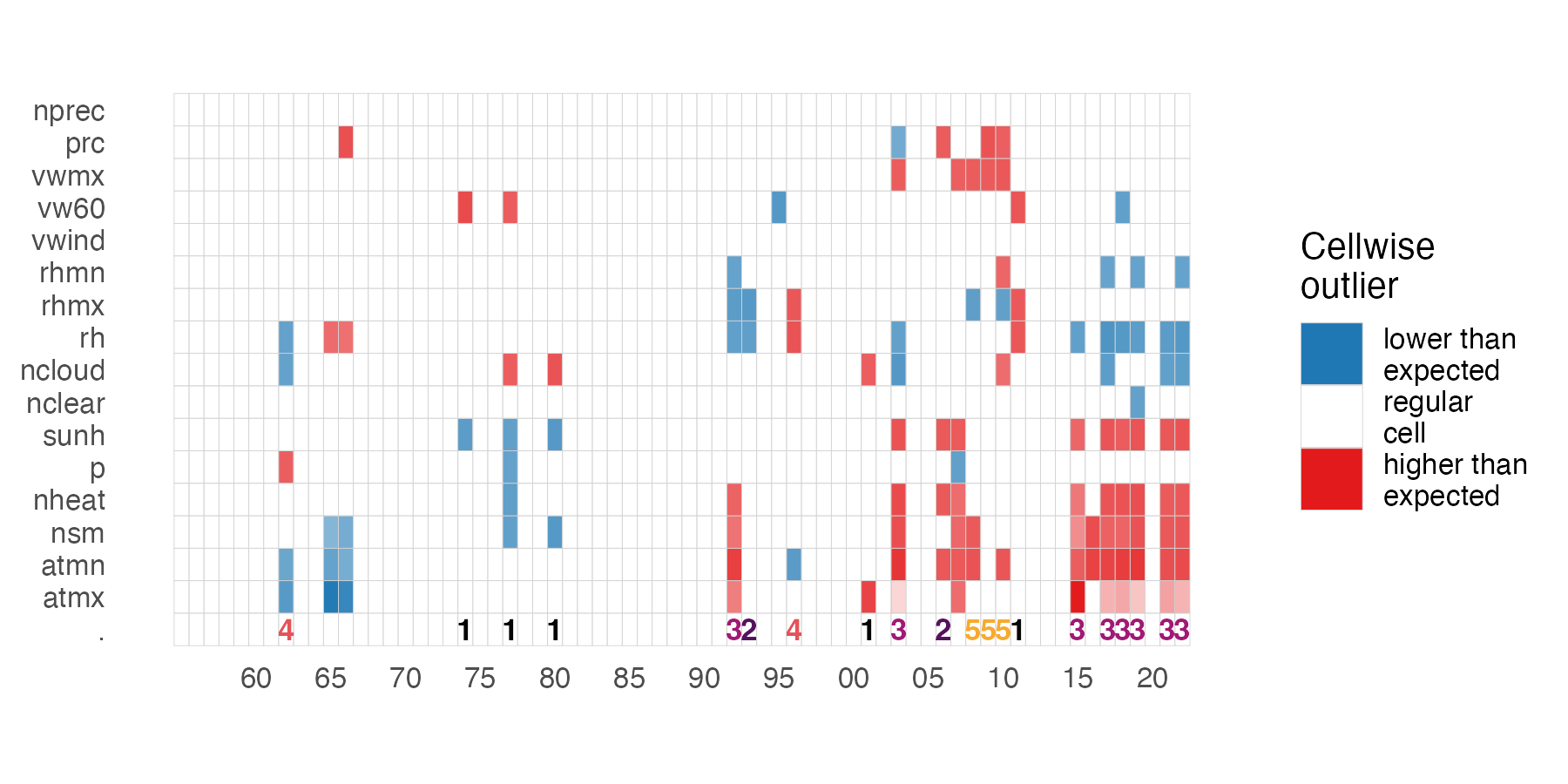}

}

\caption{\label{fig-shapley}Reproduction of the results from Mayrhofer
and Filzmoser (2023), showing which cells have been identified as
outlying by their Shapley Cell Detector algorithm (Mayrhofer and
Filzmoser 2023). Red cells indicate values higher than expected, blue
cells values lower than expected. For years that have been identified as
outliers in our approach we show the cluster assignment in the bottom
row.}

\end{figure}%

\section{Conclusion}\label{conclusion}

This paper has provided a new projection pursuit index for comparing a
sample of data against a multivariate normal reference distribution. It
has also provided an analytical result for drawing \gD{} projections of
the \pD{} ellipse corresponding to the normal reference distribution.
These combined provide new ways to visualise multivariate data, for this
particular scenario.

The work is related to outlier detection methods (see, Mayrhofer and
Filzmoser 2023, for example), particularly those that use robust
statistics, as illustrated in the weather example. Methods discussed in
Rousseeuw, Raymaekers, and Hubert (2018) and Donoho and Gasko (1992)
describe detecting outlyingness using projections of the data. When
these projections are collected and used as a large set, they define a
high-dimensional confidence region of an indeterminate shape. There is
no analytical definition. However, we could explore the data relative to
the high-dimensional shapes using the same approach described in
Section~\ref{sec-background}. You would generate points on the surface
of the irregular shape, overlay the data on this to visualise with a
tour.

The methods on directional outlyingness or Stahel-Donoho outlyingness,
described above, lend themselves to new potential \bD{} projection
pursuit indexes. These would integrate nicely into a \bD{} projection
pursuit guided tour to interactively visualise the potential outliers
relative to the rest of the sample.

Another potential direction of this work is with model-based clustering
using Gaussian mixtures (Fraley and Raftery 2002), where ellipses form
the basis of the model. Currently, the \texttt{mclust} software (Scrucca
et al. 2023) has the capacity to show \gD{} ellipses for two variables,
or axis parallel \gD{} ellipses for \pD{}. The equations developed here
for drawing the \gD{} projection of a \pD{} ellipse would provide more
versatile visualisation for examining the model-based clustering fits.

\section*{Acknowledgements}\label{acknowledgements}
\addcontentsline{toc}{section}{Acknowledgements}

The work of first author Calvi was supported by the ResearchFirst
undergraduate research program at Monash University. German Valencia
advised on the project, in relation to physics applications.

\section*{Supplementary Materials}\label{supplementary-materials}
\addcontentsline{toc}{section}{Supplementary Materials}

The full paper, code and data to reproduce the work, are publicly
available at https://github.com/uschiLaa/anomaly\_ppi.

\section*{References}\label{references}
\addcontentsline{toc}{section}{References}

\phantomsection\label{refs}


\begin{CSLReferences}{1}{0}
\bibitem[\citeproctext]{ref-As85}
Asimov, Daniel. 1985. {``{T}he {G}rand {T}our: {A} {T}ool for {V}iewing
{M}ultidimensional {D}ata.''} \emph{SIAM Journal of Scientific and
Statistical Computing} 6 (1): 128--43.

\bibitem[\citeproctext]{ref-BA86b}
Buja, Andreas, and Daniel Asimov. 1986. {``{G}rand {T}our {M}ethods:
{A}n {O}utline.''} \emph{Computing Science and Statistics} 17: 63--67.

\bibitem[\citeproctext]{ref-BCAH05}
Buja, Andreas, Dianne Cook, Daniel Asimov, and Catherine Hurley. 2005.
{``{C}omputational {M}ethods for {H}igh-{D}imensional {R}otations in
{D}ata {V}isualization.''} In \emph{Handbook of Statistics: Data Mining
and Visualization}, edited by C. R. Rao, E. J. Wegman, and J. L. Solka,
391--414. Amsterdam, The Netherlands: Elsevier/North-Holland.

\bibitem[\citeproctext]{ref-cook1995}
Cook, Dianne, Andreas Buja, Javier Cabrera, and Catherine Hurley. 1995.
{``Grand Tour and Projection Pursuit.''} \emph{Journal of Computational
and Graphical Statistics} 4 (3): 155--72.
\url{https://doi.org/10.1080/10618600.1995.10474674}.

\bibitem[\citeproctext]{ref-CLBW06}
Cook, Dianne, Eun-Kyung Lee, Andreas Buja, and Hadley Wickham. 2006.
{``{G}rand {T}ours, {P}rojection {P}ursuit {G}uided {T}ours and {M}anual
{C}ontrols.''} In \emph{Handbook of {D}ata {V}isualization}, edited by
C.-H. Chen, W. Härdle, and A. Unwin. Berlin: Springer.

\bibitem[\citeproctext]{ref-CS07}
Cook, Dianne, and Deborah F. Swayne. 2007. \emph{Interactive and Dynamic
Graphics for Data Analysis: With {R} and {GGobi}}. Use r! New York:
Springer-Verlag. \url{https://doi.org/10.1007/978-0-387-71762-3}.

\bibitem[\citeproctext]{ref-donoho1982}
Donoho, David L., and Miriam Gasko. 1992. {``{Breakdown Properties of
Location Estimates Based on Halfspace Depth and Projected
Outlyingness}.''} \emph{The Annals of Statistics} 20 (4): 1803--27.
\url{https://doi.org/10.1214/aos/1176348890}.

\bibitem[\citeproctext]{ref-FR02}
Fraley, Chris, and Adrian Raftery. 2002. {``Model-Based {C}lustering,
{D}iscriminant {A}nalysis, {D}ensity {E}stimation.''} \emph{Journal of
the American Statistical Association} 97: 611--31.

\bibitem[\citeproctext]{ref-FT74}
Friedman, Jerome H., and John W. Tukey. 1974. {``{A} {P}rojection
{P}ursuit {A}lgorithm for {E}xploratory {D}ata {A}nalysis.''} \emph{IEEE
Transactions on Computing C} 23: 881--89.

\bibitem[\citeproctext]{ref-Ga71}
Gabriel, K. Ruben. 1971. {``The {B}iplot {G}raphical {D}isplay of
{M}atrices with {A}pplications to {P}rincipal {C}omponent {A}nalysis.''}
\emph{Biometrika} 58: 453--67.

\bibitem[\citeproctext]{ref-GH96}
Gower, John C., and David J. Hand. 1996. \emph{Biplots}. London:
Chapman; Hall.

\bibitem[\citeproctext]{ref-dunn}
Halkidi, Maria, Yannis Batistakis, and Michalis Vazirgiannis. 2002.
{``Clustering Validity Checking Methods: Part II.''} \emph{ACM SIGMOD
Record} 31 (September). \url{https://doi.org/10.1145/601858.601862}.

\bibitem[\citeproctext]{ref-fpc}
Hennig, Christian. 2024. \emph{Fpc: Flexible Procedures for Clustering}.
\url{https://CRAN.R-project.org/package=fpc}.

\bibitem[\citeproctext]{ref-Hu85}
Huber, Peter J. 1985. {``{P}rojection {P}ursuit (with Discussion).''}
\emph{Annals of Statistics} 13: 435--525.

\bibitem[\citeproctext]{ref-JS87}
Jones, M. C., and Robin Sibson. 1987. {``{W}hat Is {P}rojection
{P}ursuit? (With Discussion).''} \emph{Journal of the Royal Statistical
Society, Series A} 150: 1--36.

\bibitem[\citeproctext]{ref-Kandanaarachchi02012021}
Kandanaarachchi, Sevvandi, and Rob J. Hyndman. 2021. {``Dimension
Reduction for Outlier Detection Using DOBIN.''} \emph{Journal of
Computational and Graphical Statistics} 30 (1): 204--19.
https://doi.org/\url{https://doi.org/10.1080/10618600.2020.1807353}.

\bibitem[\citeproctext]{ref-lib-med-liver-norms}
Lala, Vasimahmed, Muhammad Zubair, and David A. Minter. 2023. {``Liver
Function Tests.''} \url{https://www.ncbi.nlm.nih.gov/books/NBK482489/}.

\bibitem[\citeproctext]{ref-tours2022}
Lee, Stuart, Dianne Cook, Natalia da Silva, Ursula Laa, Nicholas
Spyrison, Earo Wang, and H. Sherry Zhang. 2022. {``The State-of-the-Art
on Tours for Dynamic Visualization of High-Dimensional Data.''}
\emph{WIREs Computational Statistics} 14 (4): e1573.
\url{https://doi.org/10.1002/wics.1573}.

\bibitem[\citeproctext]{ref-loperfido}
Loperfido, Nicola. 2018. {``Skewness-Based Projection Pursuit: A
Computational Approach.''} \emph{Computational Statistics \& Data
Analysis} 120: 42--57.
https://doi.org/\url{https://doi.org/10.1016/j.csda.2017.11.001}.

\bibitem[\citeproctext]{ref-robustbase}
Maechler, Martin, Peter J. Rousseeuw, Christophe Croux, Valentin
Todorov, Andreas Ruckstuhl, Matias Salibian-Barrera, Tobias Verbeke,
Manuel Koller, Eduardo L. T. Conceicao, and Maria Anna di Palma. 2024.
\emph{Robustbase: Basic Robust Statistics}.
\url{http://robustbase.r-forge.r-project.org/}.

\bibitem[\citeproctext]{ref-mahalanobis}
Mahalanobis, Prasanta Chandra. 1936. {``On the Generalised Distance in
Statistics.''} \emph{Sankhya A} 80 (Suppl 1): 1--7.
\url{https://doi.org/10.1007/s13171-019-00164-5}.

\bibitem[\citeproctext]{ref-MAYRHOFER2023}
Mayrhofer, Marcus, and Peter Filzmoser. 2023. {``Multivariate Outlier
Explanations Using Shapley Values and Mahalanobis Distances.''}
\emph{Econometrics and Statistics}.
\url{https://doi.org/10.1016/j.ecosta.2023.04.003}.

\bibitem[\citeproctext]{ref-rousseeuw1985multivariate}
Rousseeuw, Peter J. 1985. {``Multivariate Estimation with High Breakdown
Point.''} \emph{Mathematical Statistics and Applications} 8 (283-297):
37.

\bibitem[\citeproctext]{ref-directional-outlyingness}
Rousseeuw, Peter J., Jakob Raymaekers, and Mia Hubert. 2018. {``A
Measure of Directional Outlyingness with Applications to Image Data and
Video.''} \emph{Journal of Computational and Graphical Statistics} 27
(2): 345--59. \url{https://doi.org/10.1080/10618600.2017.1366912}.

\bibitem[\citeproctext]{ref-mclust}
Scrucca, Luca, Chris Fraley, T. Brendan Murphy, and Adrian E. Raftery.
2023. \emph{Model-Based Clustering, Classification, and Density
Estimation Using {mclust} in {R}}. Chapman; Hall/CRC.
\url{https://doi.org/10.1201/9781003277965}.

\bibitem[\citeproctext]{ref-tourr}
Wickham, Hadley, Dianne Cook, Heike Hofmann, and Andreas Buja. 2011.
{``Tourr: {An} {R} {Package} for {Exploring} {Multivariate} {Data} with
{Projections}.''} \emph{Journal of Statistical Software} 40 (2).
\url{https://doi.org/10.18637/jss.v040.i02}.

\bibitem[\citeproctext]{ref-tourr-cran}
Wickham, Hadley, Dianne Cook, Nicholas Spyrison, Ursula Laa, H. Sherry
Zhang, and Stuart Lee. 2024. \emph{Tour Methods for Multivariate Data
Visualisation}. \url{https://CRAN.R-project.org/package=tourr}.

\bibitem[\citeproctext]{ref-ferrn}
Zhang, H. Sherry, Dianne Cook, Ursula Laa, Nicolas Langrené, and
Patricia Menéndez. 2021. {``Visual Diagnostics for Constrained
Optimisation with Application to Guided Tours.''} \emph{The R Journal}
13 (2): 624--41. \url{https://doi.org/10.32614/RJ-2021-105}.

\end{CSLReferences}
\end{document}